\documentclass[final,leqno]{siamltex}
\usepackage[dvips]{graphicx,color}
\usepackage{amsmath,amssymb}

\title{One-Dimensional Tunnel-Junction Formula 
for Schr\"{o}dinger Particle}

\author{Masao Hirokawa\thanks{Department of Mathematics, Okayama University, 
Okayama, 700-8530, Japan ({\tt hirokawa@math.okayama-u.ac.jp}).}
        \and Takuya Kosaka\thanks{Department of Mathematics, 
Okayama University, Okayama, 700-8530, Japan.}}

\begin{document}

\maketitle

\begin{abstract}
We handle all the self-adjoint extensions 
of the minimal Schr\"{o}dinger operator 
for the non-relativistic electron 
living in the one-dimensional configuration space 
with a junction. 
We are interested in every boundary condition 
corresponding to the individual self-adjoint extension. 
Thus, we clarify all the types of those boundary conditions 
of the wave functions of the non-relativistic electron. 
We find a tunnel-junction formula 
for the non-relativistic electron 
passing through the junction. 
Using this tunnel-junction formula, 
we propose a mathematical possibility 
of a tunnel-junction device for qubit. 
\end{abstract}



\pagestyle{myheadings}
\thispagestyle{plain}
\markboth{M. HIROKAWA AND T. KOSAKA}{One-Dimensional Tunnel-Junction Formula}

\section{Introduction} 

In this paper we consider a single electron 
living in the one-dimensional configuration space with a junction. 
Thus, we consider it the Schr\"{o}dinger particle.  
We mathematically simplify our physical set-up in the following: 
We regard the one-dimensional configuration space with 
the junction as the space $\Omega_{\Lambda}:=
\mathbb{R}\setminus [-\Lambda, +\Lambda]$. 
The segment $[-\Lambda, +\Lambda]$ is the junction then. 
We have been interested in the boundary condition 
of the wave functions of the non-relativistic electron 
so that the electron's energy operator 
(i.e., the Schr\"{o}dinger operator) 
becomes an observable (i.e., self-adjoint).  
We have considered some self-adjoint extensions of 
the minimal Schr\"{o}dinger operator $H_{0}$ 
for this configuration space in \cite{FHNS10,SH11} 
from the point of the view of a quantum device. 
We then showed that there are some cases where 
the wave functions on which a self-adjoint extension acts 
have their own phase factor at the boundary, $\pm\Lambda$, 
when the wave functions pass through the junction. 
On the other hand, we showed that 
the wave functions do not have such a phase factor 
when they respectively stay in the left island 
$\Omega_{\Lambda,L}:=(-\infty, -\Lambda)$ 
and in the right island  
$\Omega_{\Lambda,R}:=(+\Lambda, +\infty)$, 
namely, when there is no exchange between the wave functions 
living individual islands. 
Actually, all the self-adjoint extensions 
of our minimal Schr\"{o}dinger operator $H_{0}$ 
can be parameterized by 
$(u_{jk})_{j,k=1,2}\in U(2)$ following the von Neumann's theory. 
Here, $U(2)$ is the unitary group of the degree $2$. 
In \cite{FHNS10} the appearance of the phase factor was shown 
only in the cases where diagonal entries of 
the unitary matrix are all zero 
(i.e., $u_{11}=u_{22}=0$). 
In this paper, we complete the result in \cite{FHNS10}. 
As asserted in Theorem \ref{theo:main2} below, 
we will characterize the boundary conditions of all 
the self-adjoint extensions of the minimal Schr\"{o}dinger 
operator $H_{0}$ by just two types of 
boundary conditions proposed in \cite{FHNS10}. 
More precisely, we will actually construct the boundary condition 
from every matrix of $U(2)$.  
We will then give a tunnel-junction formula concerning 
the phase factor for the non-relativistic electron 
as the Schr\"{o}dinger particle. 
We will propose a mathematical idea 
to make a qubit from a Schr\"{o}dinger particle 
through the tunnel-junction formula 
by controlling the phase factor, 
even though the spin of the Schr\"{o}dinger particle 
cannot be used. 
In the case where the electron spin should be considered, 
we studied similar problem for a single relativistic electron 
as the Dirac particle \cite{HK1}.   

\section{Main Theorem} 

First up, we prepare some mathematical tools 
and recall some results to state our main theorem, 
Theorem \ref{theo:main2}. 

We respectively define function spaces $AC^{2}(\Omega_{\Lambda})$ and 
$AC_{0}^{2}(\Omega_{\Lambda})$ for our configuration space 
$\Omega_{\Lambda}$ by 
\begin{align*}
AC^{2}\bigl(\overline{\Omega_{\Lambda}}\big) :=
&\Biggl\{ f \in L^{2}(\Omega_{\Lambda})\, \bigg| \, 
\text{$f$ and $f{\,}'$ are absolutely continuous on 
$\overline{\Omega_{\Lambda}}$,} \\ 
&\qquad\qquad\qquad\quad 
\text{and}\,\,\, f{\,}', f{\,}'' \in L^{2}(\Omega_{\Lambda}) 
\Biggr\},  
\end{align*}
and 
$$
AC_{0}^{2}\bigl(\overline{\Omega_{\Lambda}}\bigr):= 
\left\{ f \in AC^{2}\bigl(\overline{\Omega_{\Lambda}}\bigr)\, 
\bigg|\, 
f=f{\,}'=0\,\,\, \text{on}\,\,\, \partial \Omega_{\Lambda}
\right\}.
$$
Here $L^{2}(\Omega_{\Lambda})$ is the set of all 
the square Lebesgue-integrable functions on $\Omega_{\Lambda}$, 
and $\partial\Omega_{\Lambda}$ denotes the boundary 
$\{ -\infty, -\Lambda, +\Lambda, +\infty\}$ 
of $\Omega_{\Lambda}$. 
Thus, $f\in AC_{0}^{2}\bigl(\overline{\Omega_{\Lambda}}\bigr)$ 
satisfies the boundary condition: 
$\lim_{x\to a}f^{(n)}(x)=0$ for 
$a\in\partial\Omega_{\Lambda}$ 
and $n=0, 1$. 

We denote by $D(A)$ the domain of a linear operator $A$ 
throughout this paper. 
Now we introduce our starting object: 
\begin{definition}
{\rm (Minimal Schr\"{o}dinger Operator): 
The $1$-dimensional \textit{Schr\"{o}dinger operator} 
$H_{0}$ is defined by
$$
\begin{cases}
D(H_{0}):=AC_{0}^{2}(\overline{\Omega_{\Lambda}}), \\ 
H_{0}:=\, -\, 
{\displaystyle \frac{d^{2}}{dx^{2}}}.  
\end{cases}
$$
We call the operator $H_{0}$ the \textit{minimal Schr\"{o}dinger operator}.
}  
\end{definition}

Following \cite[Theorems 8.25(b) and 8.22]{weidmann} respectively, 
we realize that the minimal Schr\"{o}dinger operator 
$H_{0}$ is closed symmetric, and its adjoint operator $H_{0}^{*}$ 
is given as
$$
\begin{cases}
D(H_{0}^{*}) = AC^{2}(\overline{\Omega_{\Lambda}}), \\ 
H_{0}^{*}=\, -\, 
{\displaystyle \frac{d^{2}}{dx^{2}}}.   
\end{cases}
$$
Then, the operation of $H_{0}^{*}$ is the same as that of $H_{0}$ 
though their domains are different from each other. 
So, we call the adjoint operator $H_{0}^{*}$ 
the \textit{maximal Schr\"{o}dinger operator}, 
and moreover, we employ the following naming: 
\begin{definition}
For every subspace $\mathcal{D}$ with the condition, 
$D(H_{0})\subset\mathcal{D}\subset D(H_{0}^{*})$, 
we call the restriction, $H_{0}^{*}\lceil\mathcal{D}$, 
of $H_{0}^{*}$ on $\mathcal{D}$ 
the \textit{Schr\"{o}dinger operator}. 
\end{definition} 

We will investigate every boundary condition of 
the wave functions on which individual 
self-adjoint extension acts. 

It was proved in \cite{FHNS10} 
that both deficiency indices are $2$: 
$n_{+}(H_{0}) = n_{-}(H_{0}) = 2$, 
where the deficiency index $n_{\pm}(H_{0})$ is defined as 
the dimension of the individual deficiency subspace 
$\mathcal{K}_{\pm}(H_{0}):= \mathrm{ker}
(\pm i-H_{0}^{*})$: 
$n_{\pm}(H_{0}):=\mathrm{dim}\,\mathcal{K}_{\pm}(H_{0})$. 
Thus, the minimal Schr\"{o}dinger operator has uncountably many 
self-adjoint extensions. 
More precisely, following the general theory of differential equation 
and solving simple differential equations: 
$H_{0}^{*}\psi=\pm i\psi$, 
we can obtain orthonormal bases, 
$\left\{ L_{\pm} , R_{\pm}\right\}$, 
of the deficiency subspaces 
$\mathcal{K}_{\pm}(H_{0})$, respectively:
\begin{eqnarray}
&
L_{+}(x):= \left\{ 
\begin{array}{cl}
Ne^{(1-i)x/\sqrt{2}} & \textrm{if} -\infty<x<\Lambda, \\ 
0 & \textrm{if} \ \Lambda<x<\infty, \\ 
\end{array} \right. \\ 
&
L_{-}(x):= \left\{ 
\begin{array}{cl}
Ne^{(1+i)x/\sqrt{2}} & \textrm{if} \ -\infty<x<\Lambda, \\ 
0 & \textrm{if} \ \Lambda<x<\infty, \\ 
\end{array} \right.  
\end{eqnarray}
and 
\begin{eqnarray}
&
R_{+}(x):= \left\{ 
\begin{array}{cl}
0 & \textrm{if} \ -\infty< x < \Lambda, \\ 
Ne^{(-1+i)x/\sqrt{2}} & \textrm{if} \ \Lambda<x<\infty, \\ 
\end{array} \right. \\ 
&
R_{-}(x):= \left\{ 
\begin{array}{cl}
0 & \textrm{if} \ -\infty<x<\Lambda, \\ 
Ne^{(-1-i)x/\sqrt{2}} & \textrm{if} \ \Lambda<x<\infty, \\ 
\end{array} \right.
\end{eqnarray}
with the normalization factor $N={}^{4}\!\!\!\sqrt{2}e^{\Lambda/\sqrt{2}}$ 
so that  
$H_{0}^{*}R_{\pm}=\pm iR_{\pm}$ 
and  
$H_{0}^{*}L_{\pm}=\pm iL_{\pm}$. 
The uniqueness of the differential equations tells us that 
the individual bases are complete and  
$n_{+}(H_{0}) = n_{-}(H_{0}) = 2$. 
We here note the following relations to use later: 
\begin{equation}
R_{\pm}'(x)=\frac{-1\pm i}{\sqrt{2}}R_{\pm}(x)\quad 
\text{and}\quad 
L_{\pm}'(x)=\frac{1\mp i}{\sqrt{2}}L_{\pm}(x), 
\label{eq:rel1}
\end{equation}
and 
\begin{equation}
\begin{cases}
{\displaystyle L_{+}(-\Lambda)=R_{+}(+\Lambda)=
Ne^{(-1+i)\Lambda/\sqrt{2}}}, \\ 
{\displaystyle L_{-}(-\Lambda)=R_{-}(+\Lambda)=
R_{+}(+\Lambda)^{*}=
Ne^{(-1-i)\Lambda/\sqrt{2}}}, \\ 
R_{+}(+\Lambda)^{*}=R_{+}(+\Lambda)e^{-i\sqrt{2}\, \Lambda}.
\end{cases}
\label{eq:rel2}
\end{equation}

Following the von Neumann's theory \cite{RS2,weidmann}, 
all the self-adjoint extensions $H_{\mathrm{sa}}$ of the minimal Schr\"{o}dinger 
operator $H_{0}$ are given as a restriction of 
the maximal Schr\"{o}dinger operator $H_{0}^{*}$ 
on a proper subspace $\mathcal{D}_{\mathrm{sa}}$ 
with $D(H_{0})\subset\mathcal{D}_{\mathrm{sa}}\subset D(H_{0}^{*})$: 
$H_{\mathrm{sa}}=H_{0}^{*}\lceil \mathcal{D}_{\mathrm{sa}}$. 
Then, von Neumann's theory \cite{RS2,weidmann} 
provides the following proposition: 
\begin{proposition}
\label{prop:von-Neumann}
There is a one-to-one correspondence between 
self-adjoint extensions $H_{\mathrm{sa}}$ of 
the minimal Schr\"{o}dinger operator $H_{0}$ 
and unitary operators $U : \mathcal{K}_{+}(H_{0}) 
\to \mathcal{K}_{-}(H_{0})$ so that the correspondence is 
given in the following: 
For every unitary operator $U : 
\mathcal{K}_{+}(H_{0}) \to \mathcal{K}_{-}(H_{0})$, 
the corresponding self-adjoint extension $H_{U}$ 
is defined by 
\begin{equation}
\begin{cases}
D(H_{U}):=\left\{\psi = \psi_{0} + \psi^{+} + U\psi^{+}\, 
\big|\, \psi_{0} \in D(H_{0}),\, \psi^{+} \in \mathcal{K}_{+}(H_{0})
\right\}, \\
H_{U}:=H_{0}^{*} 
\lceil{D(H_{U})}, 
\end{cases}
\label{eq:von-Neumann-representation}
\end{equation}
and then its operation is 
$$
H_{U}(\psi_{0} + \psi^{+} + U\psi^{+}) 
= H_{0}\psi_{0} + i\psi^{+} -iU\psi^{+}.
$$
Conversely, for every self-adjoint extension $H_{\mathrm{sa}}$ 
of the  minimal Schr\"{o}dinger operator $H_{0}$, 
there is the corresponding unitary operator 
$U: \mathcal{K}_{+}(H_{0}) 
\to \mathcal{K}_{-}(H_{0})$ 
so that 
$H_{\mathrm{sa}}=H_{U}$. 
\end{proposition} 

We here introduce some mathematical notation and 
terminology. 
We denote by $\overline{\mathbb{R}}$ 
the set of all extended real numbers: 
$\overline{\mathbb{R}}:=\mathbb{R}\cup\{+\infty\}$. 
Since $n_{+}(H_{0})=n_{-}(H_{0})=2$, 
the deficiency subspaces $\mathcal{K}_{\pm}(H_{0})$ are 
$2$-dimensional Hilbert spaces. 
We here remember that 
the sets $\{ L_{+}, R_{+} \}$ and $\{ L_{-}, R_{-} \}$ are respectively 
the complete orthonormal systems of the deficiency subspace 
$\mathcal{K}_{+}(H_{0})$ and $\mathcal{K}_{-}(H_{0})$. 
We identify unitary operators $U$ from 
$\mathcal{K}_{+}(H_{0})$ to $\mathcal{K}_{-}(H_{0})$ with 
$2\times 2$ unitary matrices $(u_{jk})_{j,k=1,2}$, 
making the correspondence by 
$U : \psi_{j}^{+} \longmapsto \sum_{k=1}^{2} u_{jk} \psi_{k}^{-}$, 
$j = 1, 2$, 
where $\psi_{1}^{\pm}:=L_{\pm}$ and 
$\psi_{2}^{\pm}:=R_{\pm}$. 
So, we often identify the unitary operator $U$ 
with the unitary matrix $(u_{jk})_{jk=1,2}$, 
and write $U\in  U(2)$. 
We denote by $U(n)$ the unitary group of degree $n\in\mathbb{N}$ 
throughout this paper. 
The representation of our $U:\mathcal{K}_{+}(H_{0})\to 
\mathcal{K}_{-}(H_{0})$ in this paper is then: 
\begin{equation}
\begin{cases}
UL_{+}=u_{11}L_{-}+u_{12}R_{-}, \\ 
UR_{+}=u_{21}L_{-}+u_{22}R_{-}. 
\end{cases}
\label{eq:tr-prob}
\end{equation}
We say that 
$U$ is \textit{diagonal} if 
$u_{jk} = 0$ with $j\ne k$. 
Otherwise, we say $U$ is \textit{non-diagonal}.

Before stating our main theorem, Theorem \ref{theo:main2}, 
we recall the two types of self-adjoint extensions 
found in \cite{FHNS10}. 
We introduce a class of vectors 
$\alpha=(\alpha_{1},  \alpha_{2},  
\alpha_{3},  \alpha_{4})\in\mathbb{C}^{4}$ 
\cite[Definition 1]{FHNS10}: 
We say $\alpha\in\mathbb{C}^{4}$ is in 
a class (Class $\alpha$) if and only if 
\begin{equation}
\alpha_{1}\alpha_{4}^{*}-\alpha_{2}\alpha_{3}^{*}=1 
\label{eq:lem-3-1} 
\end{equation}
and 
\begin{equation}
\alpha_{1}\alpha_{3}^{*},\,\,\, \alpha_{2}\alpha_{4}^{*}\in\mathbb{R}.
\label{eq:lem-3-2} 
\end{equation}
For every vector $\alpha$ in the class (Class $\alpha$), 
we give a matrix $B_{\alpha}$ by 
${\displaystyle 
B_{\alpha}:=
\begin{pmatrix}
\alpha_{1} & \alpha_{2} \\ 
\alpha_{3} & \alpha_{4}   
\end{pmatrix}
}$, 
and call it the \textit{boundary matrix}. 
We recall that there are at least two types of 
boundary conditions for self-adjoint extensions 
of the the minimal Schr\"{o}dinger operator $H_{0}$ 
as shown in \cite{FHNS10}.  
We will give a brief outline of its proof 
in \S\ref{subsec:theo-main1}. 
\begin{theorem}
\label{theo:main1}
(\cite[Theorem 1(ii){\&}Theorem 2(i)]{FHNS10})
\begin{enumerate}
\item[(a)] For every $\rho=(\rho_{+} , \rho_{-})
\in \overline{\mathbb{R}}^{2}$, set a subspace $D(H_{\rho})$ as  
$$
D(H_{\rho}):=\left\{\psi\in D(H_{0}^{*}) \,|\, 
\psi\,\, \text{satisfies the boundary condition 
(\ref{schrodinger:*rho})} \right\}, 
$$
where 
\begin{equation}
\begin{cases}
\rho_{+}\psi(+\Lambda)= \psi{\,}'(+\Lambda) 
& \text{if\,\,\, $|\rho_{+}|<\infty$}, \\ 
\psi(+\Lambda)=0 
& \text{if\,\,\, $\rho_{+}=\infty$},  \\ 
\rho_{-}\psi(-\Lambda)=\psi{\,}'(-\Lambda) 
&\text{if\,\,\, $|\rho_{-}|<\infty$}, \\ 
\psi(-\Lambda)=0 
&\text{if\,\,\, $\rho_{-}=\infty$}. 
\end{cases}
\tag{BC $\rho$}
\label{schrodinger:*rho}
\end{equation}
Then, the restriction, $H_{\rho}:=H_{0}^{*}\lceil D(H_{\rho})$, 
of the adjoint operator $H_{0}^{*}$ on $D(H_{\rho})$ is 
a self-adjoint extension of the minimal Schr\"{o}dinger 
operator $H_{0}$. 
\item[(b)] For every vector $\alpha$ in the class (Class $\alpha$), 
define a subspace $D(H_{\alpha})$ by 
$$
D(H_{\alpha})=\biggl\{
\psi\in D(H_{0}^{*})\, \biggl| \, 
\psi\,\, \text{satisfies the boundary condition 
(\ref{schrodinger:*alpha})} \biggl\}, 
$$
where  
\begin{equation}
\begin{pmatrix}
\psi(+\Lambda) \\ 
\psi{\,}'(+\Lambda)
\end{pmatrix}
= B_{\alpha}
\begin{pmatrix}
\psi(-\Lambda) \\ 
\psi{\,}'(-\Lambda) 
\end{pmatrix}.
\tag{BC $\alpha$}
\label{schrodinger:*alpha}
\end{equation}
Then, the restriction, $H_{\alpha}:=H_{0}^{*}\lceil D(H_{\alpha})$, 
of the adjoint operator $H_{0}^{*}$ on $D(H_{\alpha})$ 
is a self-adjoint extension of the minimal Schr\"{o}dinger 
operator $H_{0}$. 
\end{enumerate}
\end{theorem}

We here recall the following lemma:
\begin{lemma}(\cite[Lemma 2]{FHNS10})
\label{lem:1}
If $\alpha_{1}, \alpha_{2}, \alpha_{3}, \alpha_{4} 
\in\mathbb{C}$ in in the class (Class $\alpha$), 
then $\alpha_{j}\alpha_{k}^{*} \in\mathbb{R}$ 
for each $j, k=1, 2, 3, 4$. 
\end{lemma}
\begin{proof} 
We only have to show the case where $j\ne j^{\prime}$. 
Multiplying both sides of (\ref{eq:lem-3-1}) by $\alpha_{3}^{*}$, 
we have 
\begin{equation}
\alpha_{3}^{*}=\alpha_{1}\alpha_{3}^{*}\alpha_{4}^{*}-\alpha_{2}^{*}|\alpha_{3}|^{2}.
\label{eq:lem-3-3}
\end{equation} 
Multiply both sides of this equation by $\alpha_{2}$. 
Then, (\ref{eq:lem-3-2}) tells us that $\alpha_{2}\alpha_{3}^{*}\in\mathbb{R}$. 
Combining this fact with (\ref{eq:lem-3-1}), 
we have 
$\alpha_{1}\alpha_{4}^{*}
=1 + \alpha_{2}\alpha_{3}^{*}\in \mathbb{R}$. 
Multiplying both sides of (\ref{eq:lem-3-1}) by $\alpha_{2}$ 
leads to $\alpha_{2}
=\alpha_{1}\alpha_{2}\alpha_{4}^{*}-|\alpha_{2}|^{2}\alpha_{3}$. 
Here we used $\alpha_{2}^{*}\alpha_{3}
=\alpha_{2}\alpha_{3}^{*}\in\mathbb{R}$. 
Multiplying both sides of the above representation of $\alpha_{2}$ 
by $\alpha_{1}^{*}$ brings us to $\alpha_{1}^{*}\alpha_{2}\in\mathbb{R}$. 
Here we used conditions in (\ref{eq:lem-3-2}), 
especially, $\alpha_{1}^{*}\alpha_{3}
=\alpha_{1}\alpha_{3}^{*}\in\mathbb{R}$. 
We reach 
$\alpha_{3}^{*}\alpha_{4}= \alpha_{1}\alpha_{3}^{*}|\alpha_{4}|^{2}
-\alpha_{2}^{*}\alpha_{4}|\alpha_{3}|^{2} \in \mathbb{R}$ 
by multiplying both sides of (\ref{eq:lem-3-3}) by $\alpha_{4}$ 
and using the conditions in (\ref{eq:lem-3-2}). 
Here we used $\alpha_{2}^{*}\alpha_{4}
=\alpha_{2}\alpha_{4}^{*}\in\mathbb{R}$. 
We have proved our statement for $\alpha_{2}\alpha_{3}^{*}$, 
$\alpha_{1}\alpha_{4}^{*}$, $\alpha_{1}^{*}\alpha_{2}$, 
and $\alpha_{3}^{*}\alpha_{4}$. 
It follows from these facts and the conditions in (\ref{eq:lem-3-2}) 
that $\alpha_{j}\alpha_{k}^{*}\in\mathbb{R}$ 
for all the combinations of $j, k= 1, 2, 3, 4$. 
\end{proof}

The following proposition shows how a phase factor appears 
in the boundary matrices $B_{\alpha}$ 
with $\alpha\in\mathbb{C}^{4}$ 
in the class (Class $\alpha$): 
\begin{proposition}
\label{prop:phase}
Let $B_{\alpha}$ be an arbitrary boundary matrix 
with the vector $\alpha=(\alpha_{1},\alpha_{2},\alpha_{3},\alpha_{4})
\in\mathbb{C}^{4}$ in the class (Class $\alpha$). 
Then, one of $\alpha_{1}$ and $\alpha_{2}$ is non-zero at least. 
So, set $\theta\in \left[\left. 0 , 2\pi\right)\right.$, and 
$a_{1}, a_{2}, a_{3}, a_{4}\in\mathbb{R}$ as follows: 
Let $\alpha_{j}$ be $\alpha_{1}$ if $\alpha_{1}\ne 0$, 
and $\alpha_{2}$ if $\alpha_{1}=0$. 
Define
$$
\begin{cases}
\theta:= \arg(\alpha_{j}/|\alpha_{j}|); \\ 
a_{j}:= |\alpha_{j}|,\,\,\, 
a_{k}:= \alpha_{k}\alpha_{j}^{*}/|\alpha_{j}|,\,\,\, 
k\ne j. 
\end{cases}
$$
Then, $B_{\alpha}$ has the following representation: 
$$
B_{\alpha}=
e^{i\theta}
\begin{pmatrix}
a_{1} & a_{2} \\ 
a_{3} & a_{4}
\end{pmatrix}.
$$
\end{proposition}

\begin{proof}
Set $\theta_{j}$ as $\theta_{j}= \arg\alpha_{j}$. 
Since the vector $\alpha$ is in the class (Class $\alpha$), 
Lemma \ref{lem:1} says that $\alpha_{j}\alpha_{k}^{*}\in\mathbb{R}$ 
for each $j, k= 1, 2, 3, 4$. 
Moreover,  (\ref{eq:lem-3-1}) says 
that $\alpha_{1}\ne 0$ or $\alpha_{2}\ne 0$. 
We can rewrite $B_{\alpha}$ as:
$$
B_{\alpha} = \frac{\alpha_{1}}{|\alpha_{1}|}
\begin{pmatrix}
{\displaystyle |\alpha_{1}|} & 
{\displaystyle \frac{\alpha_{1}^{*}\alpha_{2}}{|\alpha_{1}|}} \\ 
\quad & \quad \\  
{\displaystyle \frac{\alpha_{1}^{*}\alpha_{3}}{|\alpha_{1}|}} & 
{\displaystyle \frac{\alpha_{1}^{*}\alpha_{4}}{|\alpha_{1}|}}
\end{pmatrix}
\,\,\, 
\text{in the case where $\alpha_{1}\ne 0$,}
$$ 
and 
$$
B_{\alpha} = \frac{\alpha_{2}}{|\alpha_{2}|}
\begin{pmatrix}
{\displaystyle \frac{\alpha_{1}\alpha_{2}^{*}}{|\alpha_{2}|}} & 
{\displaystyle |\alpha_{2}|} \\ 
\quad & \quad \\  
{\displaystyle \frac{\alpha_{3}\alpha_{2}^{*}}{|\alpha_{2}|}} & 
{\displaystyle \frac{\alpha_{4}\alpha_{2}^{*}}{|\alpha_{2}|}}
\end{pmatrix} 
\,\,\, 
\text{in the case where $\alpha_{2}\ne 0$}.
$$ 
We note 
$a_{1}a_{4}-a_{2}a_{3}=\alpha_{1}\alpha_{4}^{*}-\alpha_{2}\alpha_{3}^{*}=1$ 
by (\ref{eq:lem-3-1}). 
Therefore, we can complete our proof. 
\end{proof}

Our main theorem in this paper is concerned with 
the classification of all the boundary conditions. 
To state the classification, 
we prepare another lemma. 
We denote by $S\mathbb{H}$ the Hamilton quaternion field 
with determinant one, i.e., 
$$
S\mathbb{H}:=
\left\{
\begin{pmatrix}
\gamma_{1} & - \gamma_{2}^{*} \\ 
\gamma_{2} & \gamma_{1}^{*} 
\end{pmatrix}\, \bigg|\, 
\gamma_{1}, \gamma_{2} \in \mathbb{C},\,\,\, 
|\gamma_{1}|^{2}+|\gamma_{2}|^{2}=1
\right\}.
$$
The following lemma says that $U\in U(2)$ can be decomposed 
into the product of an element of $U(1)$ 
and an element of $S\mathbb{H}$. 
Although this lemma was already proved in \cite[Proposition 4.3]{HK1}, 
we here give a simpler proof than that of \cite[Proposition 4.3]{HK1}: 
\begin{lemma}
\label{lem:U(2)} 
$U(2)=U(1)S\mathbb{H}
\equiv \{\gamma_{3}A\, |\, \gamma_{3}\in U(1),\, A\in S\mathbb{H}\}$. 
\end{lemma}

\begin{proof} Since it is clear that 
$U(2)\supset U(1)S\mathbb{H}$, 
we show $U(2)\subset U(1)S\mathbb{H}$. 
Let ${\displaystyle 
U=\begin{pmatrix}
   u_{1} & u_{2} \\ 
   u_{3} & u_{4}
  \end{pmatrix}}
$ be an arbitrary matrix 
in $U(2)$. 
The equation $I_{\mathbb{C}^{2}}=U^{*}U$ implies 
\begin{equation}
|u_{1}|^{2}+|u_{3}|^{2}=1. 
\label{eq:U(2)-3}
\end{equation} 
Similarly, the equation $I_{\mathbb{C}^{2}}=UU^{*}$ 
implies $|u_{1}|^{2}+|u_{2}|^{2}=1$. 
Comparing this with (\ref{eq:U(2)-3}), 
we have 
\begin{equation}
|u_{2}|=|u_{3}|. 
\label{eq:U(2)-4}
\end{equation}
In the same way, 
comparing the equation 
$|u_{3}|^{2}+|u_{4}|^{2}=1$ derived from 
$I_{\mathbb{C}^{2}}=UU^{*}$ 
with (\ref{eq:U(2)-3}) leads to 
\begin{equation}
|u_{1}|=|u_{4}|. 
\label{eq:U(2)-5}
\end{equation}
Here we introduce argument $\theta_{j}$ of $u_{j}$, i.e.,  
$\theta_{j}:=\arg u_{j}$,
and then $u_{j}=:|u_{j}|e^{i\theta_{j}}$, 
 $j=1, 2, 3, 4$. 
By using $I_{\mathbb{C}^{2}}=UU^{*}$, (\ref{eq:U(2)-4}), and 
(\ref{eq:U(2)-5}), we can reach 
\begin{equation}
|u_{1}||u_{3}|
\left(
e^{i(\theta_{1}+\theta_{4}-\theta_{2}-\theta_{3})}+1
\right)=0. 
\label{eq:U(2)-6}
\end{equation}  

In the case where $u_{3}\ne 0$, 
we set $\gamma_{1}, \gamma_{2}$, and $\gamma_{3}$ 
by $\gamma_{1}:=e^{-i(\theta_{2}+\theta_{3}+\pi)/2}u_{1}$, 
$\gamma_{2}:=e^{-i(\theta_{2}+\theta_{3}+\pi)/2}u_{3}$, 
and 
$\gamma_{1}:=e^{i(\theta_{2}+\theta_{3}+\pi)/2}$. 
Then, we have $\gamma_{1}\gamma_{3}=u_{1}$ 
and $\gamma_{2}\gamma_{3}=u_{3}$. 
By (\ref{eq:U(2)-4}) we can compute $-\gamma_{2}^{*}$ as: 
$$
-\gamma_{2}^{*}=\, -e^{i(\theta_{2}+\theta_{3}+\pi)/2}|u_{3}|e^{-i\theta_{3}}
=\, -e^{i(\theta_{2}-\theta_{3}+\pi)/2}|u_{3}|
=\, -e^{i(\pi-\theta_{2}-\theta_{3})/2}u_{2}, 
$$
which implies $-\gamma_{2}^{*}\gamma_{3}=\, -e^{i\pi}u_{2}=u_{2}$. 
Meanwhile, it follows from (\ref{eq:U(2)-6}) that 
$$
|u_{1}|e^{i(\theta_{1}+\theta_{4})}
=\, -|u_{1}|e^{i(\theta_{2}+\theta_{3})}. 
$$
Using this, 
we compute $\gamma_{1}^{*}\gamma_{3}$ 
as $\gamma_{1}^{*}\gamma_{3}=e^{i(\theta_{2}+\theta_{3}+\pi)}|u_{1}|e^{-i\theta_{1}}
=\, -e^{i\pi}|u_{1}|e^{i\theta_{4}}=u_{4}$. 
Hence it follows that $U\in S\mathbb{H}$. 

In the case $u_{3}=0$, we define $\gamma_{1}, \gamma_{2}$, 
and $\gamma_{3}$ by 
$\gamma_{1}:=e^{-i(\theta_{1}+\theta_{4})/2}u_{1}$, 
$\gamma_{2}:=0$, 
and $\gamma_{3}:=e^{i(\theta_{1}+\theta_{4})/2}$, 
and then, we have 
$\gamma_{1}\gamma_{3}=u_{1}$, 
$\gamma_{2}\gamma_{3}=0=u_{3}$. 
Since $u_{2}=0$ by (\ref{eq:U(2)-4}), 
we have $-\gamma_{2}^{*}=0=u_{2}$. 
Finally, (\ref{eq:U(2)-5}) brings us to 
the computation, 
$\gamma_{1}^{*}\gamma_{3}=
e^{i(\theta_{1}+\theta_{4})}|u_{1}|e^{-i\theta_{1}}
=e^{i\theta_{4}}|u_{1}|=u_{4}$. 
Consequently, we obtain $U\in S\mathbb{H}$. 
\end{proof}

Proposition \ref{prop:von-Neumann} based on von Neumann's 
theory says that all the self-adjoint extensions of 
the minimal Schr\"{o}dinger operator $H_{0}$ are 
parameterized by $U\in U(2)$. 
Our assertion is that there are only two types of 
boundary conditions for all the self-adjoint extensions of $H_{0}$. 
They are represented by (\ref{schrodinger:*rho}) 
and  (\ref{schrodinger:*alpha}). 
Therefore, the only thing we have to do for 
the classification of the boundary conditions 
is actually to construct the boundary condition 
(\ref{schrodinger:*rho}) or (\ref{schrodinger:*alpha}) 
from every $U\in U(2)$. 
In addition, Proposition \ref{prop:von-Neumann} says 
that a diagonal $U\in U(2)$ leaves the set of the left-island functions 
(resp. the right-island functions) invariant, 
on the other hand, a non-diagonal $U\in U(2)$ 
exchanges the left-island functions and the right-island functions. 
Theorem \ref{theo:main2} shows that this situation is reflected 
in the boundary conditions.   

\begin{theorem}
\label{theo:main2}
For every self-adjoint extension $H_{U}$, $U\in U(2)$, 
of the minimal Schr\"{o}dinger operator $H_{0}$, 
every boundary condition of the wave functions in $D(H_{U})$ 
is constructed as in (a) or (b).
\begin{enumerate}
\item[(a)] Let $U \in U(2)$ be diagonal. Then, 
$U$ is represented as 
$$
U=
\begin{pmatrix}
\gamma_{L} & 0 \\ 
0 & \gamma_{R}
\end{pmatrix}
\,\,\, \text{with}\,\,\, 
|\gamma_{L}|=|\gamma_{R}|=1. 
$$
The one-to-one correspondence between 
the self-adjoint extensions $H_{U}$ 
parameterized by diagonal matrices $U\in U(2)$ and 
the self-adjoint extensions $H_{\rho}$ 
parameterized by vectors $\rho\in\overline{\mathbb{R}}^{2}$ 
such that $D(H_{U})=D(H_{\rho})$ 
is given in the following:
\begin{enumerate}
\item[(L1)] For every $\gamma_{L} \neq -e^{i\sqrt{2}\, \Lambda}$, 
set $\theta_{L}$ as $\theta_{L}:=\arg\gamma_{L}$. 
Then, the component $\rho_{-}\in\mathbb{R}$ is given by 
$$
\rho_{-}=\, 
-\, \frac{1}{\sqrt{2}}
\left\{\tan\left(
\frac{\theta_{L}}{2}
-\,\frac{\Lambda}{\sqrt{2}}\right)-1
\right\}. 
$$
Conversely, for every $\rho_{-}\in\mathbb{R}$, 
the argument $\theta_{L}$ is determined by 
$$
\theta_{L}=2\arctan(-\sqrt{2}\rho_{-}+1)+\sqrt{2}\Lambda. 
$$
\item[(L2)] The components $\gamma_{L}=\, -e^{i\sqrt{2}\, \Lambda}$ 
and $\rho_{-}=\infty$ correspond to each other. 
\item[(R1)] For every $\gamma_{R}\neq -e^{i\sqrt{2}\, \Lambda}$, 
set $\theta_{R}$ as $\theta_{R}:=\arg\gamma_{R}$. 
Then, the component $\rho_{+}\in\mathbb{R}$ is given by 
$$
\rho_{+}=\frac{1}{\sqrt{2}}
\left\{\tan\left(
\frac{\theta_{R}}{2}
-\, \frac{\Lambda}{\sqrt{2}}\right)-1
\right\}. 
$$
Conversely, for every  $\rho_{+}\in\mathbb{R}$, 
the argument $\theta_{R}$ is determined by 
$$
\theta_{R}=2\arctan(\sqrt{2}\rho_{+}+1)+\sqrt{2}\Lambda. 
$$
\item[(R2)] The components $\gamma_{R}=\, -e^{i\sqrt{2}\, \Lambda}$ and 
$\rho_{+}=\infty$ correspond to each other.
\end{enumerate}
\item[(b)] Let $U \in U(2)$ be non-diagonal. 
Then, $U$ is represented as 
$$
U = \gamma_{3}
\begin{pmatrix}
\gamma_{1} & -\gamma_{2}^{*} \\ 
\gamma_{2} & \gamma_{1}^{*}
\end{pmatrix}
\,\,\, \text{with}\,\,\, 
|\gamma_{1}|^{2}+|\gamma_{2}|^{2}
=|\gamma_{3}|=1\,\,\, \text{and}\,\,\, 
\gamma_{2}\neq 0.  
$$
The one-to-one correspondence between 
the self-adjoint extensions $H_{U}$ 
parameterized by non-diagonal matrices $U\in U(2)$ 
and the self-adjoint extensions $H_{\alpha}$ 
parameterized by vectors $\alpha\in\mathbb{C}^{4}$ 
in the class (Class $\alpha$) 
such that $D(H_{U})=D(H_{\alpha})$ 
is given in the following:
For every triple of components 
$\gamma_{1}$, $\gamma_{2}$, and $\gamma_{3}$, 
the vector $\alpha\in\mathbb{C}^{4}$ 
in the class (Class $\alpha$) is given by  
\begin{equation}
\begin{cases}
\alpha_{1}=
i\, \sqrt{2}\,\gamma_{2}^{-1}
\left( 
\Re(e^{i\pi/4}\gamma_{1})+\Re(e^{i(\pi/4-\sqrt{2}\, \Lambda)}\gamma_{3})
\right), \\ 
\alpha_{2}=\, 
-i\, \sqrt{2}\,\gamma_{2}^{-1}
\left( 
\Re\gamma_{1}+\Re(e^{-i\sqrt{2}\, \Lambda}\gamma_{3})
\right), \\ 
\alpha_{3}=\, 
-i\, \sqrt{2}\,\gamma_{2}^{-1}
\left( 
\Re\gamma_{1}+\Re(e^{i(\pi/2-\sqrt{2}\, \Lambda)}\gamma_{3})
\right),\\ 
\alpha_{4}=
i\, \sqrt{2}\,\gamma_{2}^{-1}
\left( 
\Re(e^{-i\pi/4}\gamma_{1})+\Re(e^{i(\pi/4-\sqrt{2}\, \Lambda)}\gamma_{3})
\right).
\end{cases}
\tag{TJF-$B$}
\label{eq:TJF-B}
\end{equation}
Conversely, for every $\alpha \in\mathbb{C}^{4}$ 
satisfying (Class $\alpha$), 
the triple of the components 
$\gamma_{1}$, $\gamma_{2}$, and $\gamma_{3}$ is determined by  
\begin{equation}
\begin{cases}
\gamma_{1}=
\Gamma_{0}e^{-i\theta}
\left(
e^{i\pi/4}\alpha_{1}+\alpha_{2}+\alpha_{3}+e^{-i\pi/4}\alpha_{4}
\right), \\ 
\gamma_{2}=\, -i\sqrt{2}\Gamma_{0}e^{-i\theta}, \\ 
\gamma_{3}=\, -\Gamma_{0}e^{i(\sqrt{2}\, \Lambda-\theta)}
\left( e^{-i\pi/4}\alpha_{1}-i\alpha_{2}+\alpha_{3}+e^{-i\pi/4}\alpha_{4}
\right), \\ 
\Gamma_{0}=
\left\{
\Bigl|
e^{i\pi/4}\alpha_{1}+\alpha_{2}+\alpha_{3}+e^{-i\pi/4}\alpha_{4}
\Bigr|^{2}+2
\right\}^{1/2}.
\end{cases}
\tag{TJF-$U$}
\label{eq:TJF-U}
\end{equation}
\end{enumerate}
\end{theorem}

Theorems \ref{theo:main1} and \ref{theo:main2} 
show us the following correspondence: Every diagonal 
$U\in U(2)$ corresponds to the boundary condition 
(\ref{schrodinger:*rho}), 
and every non-diagonal $U\in U(2)$ 
to the boundary condition 
(\ref{schrodinger:*alpha}).
In addition, Theorem \ref{theo:main2} says 
that there is no boundary condition 
which makes a self-adjoint extension but 
the conditions (\ref{schrodinger:*rho}) 
and (\ref{schrodinger:*alpha}). 

By Proposition \ref{prop:von-Neumann}, 
the wave functions belonging to the domain 
of every self-adjoint extension $H_{U}$ are represented  
as in (\ref{eq:von-Neumann-representation}). 
Since $\psi_{0}(\pm\Lambda)=0$ 
for $\psi_{0}\in D(H_{0})
=AC_{0}^{2}(\overline{\Omega_{\Lambda}})$, 
the unitary operator $U\in U(2)$ gives us the information 
about how the electron reflects at the boundary and 
how it passes through the junction. 
According to (\ref{eq:tr-prob}), 
the unitary operator U maps the eigenfunction 
$L_{+}$ living in the left island 
(resp. $R_{+}$ living in the right island) 
to the eigenfunction $L_{-}$ (resp. $R_{-}$) 
staying in the same island with the probability 
$|u_{11}|^{2}$ (resp. $|u_{22}|^{2}$) and the eigenfunction 
$R_{-}$ (resp. $L_{-}$) 
coming from the opposite island with the probability 
$|u_{12}|^{2}$ (resp. $|u_{21}|^{2}$). 

Theorems \ref{theo:main1} and \ref{theo:main2} show 
how the information from $U\in U(2)$ 
reflects in the boundary conditions. 
The boundary condition (\ref{schrodinger:*rho}) 
shows the solitariness:
$$
(\text{left island})\qquad  
\begin{cases}
\rho_-\psi(-\Lambda)
=\psi'(-\Lambda) 
&\,\,\, \text{if\,\,\, $\rho_{-} \in \mathbb{R}$}, \\
\psi(-\Lambda)=0
&\,\,\, \text{if\,\,\, $\rho_{-}=\infty$},
\end{cases}
$$ 
and 
$$ 
(\text{right island})\qquad  
\begin{cases}
\rho_{+}\psi(+\Lambda)
=\psi'(+\Lambda) 
&\,\,\, \text{if\,\,\, $\rho_{+} \in \mathbb{R}$}, \\
\psi(+\Lambda)=0 
&\,\,\, \text{if\,\,\, $\rho_{+}=\infty$}.  
\end{cases} 
$$
Both of the boundary conditions in the left island 
and the right one are independent of each other, 
which makes no interchange between the information 
of the individual wave functions living in 
the left island and right one. 
In addition, no extra phase factor $\theta$ appears 
with the form $e^{i\theta}$ in this boundary condition then. 
On the other hand, the boundary condition (\ref{schrodinger:*alpha}) 
shows how the individual wave functions living in each island 
make interchange between each other at the boundary. 
Proposition \ref{prop:phase} shows 
how a phase factor appears in the boundary condition: 
\begin{equation}
\begin{pmatrix}
\psi(+\Lambda) \\ 
\psi'(+\Lambda) \\ 
\end{pmatrix} 
=
e^{i\theta}
\begin{pmatrix}
a_{1}\psi(-\Lambda)
+a_{2}\psi'(-\Lambda) \\ 
a_{3}\psi(-\Lambda)
+a_{4}\psi'(-\Lambda) 
\end{pmatrix}
\tag{TJF}
\label{eq:TJF}
\end{equation}
for some $a_{j}\in\mathbb{R}$, 
$j= 1, \cdots, 4$, with 
$a_{1}a_{4}-a_{2}a_{3}=1$. 
(\ref{eq:TJF}) with (\ref{eq:TJF-B}) and (\ref{eq:TJF-U}) 
is our \textit{tunnel-junction formula} 
for the Schr\"{o}dinger particle. 

Using this tunnel-junction formula, 
we try to consider a possibility of a quantum device. 
That is, let us now consider a unit of a quantum device, 
consisting of a junction and two quantum wires 
as in Fig.~\ref{fig:unit}.  
\begin{figure}[htbp]
\begin{center}
  \includegraphics[width=65mm]{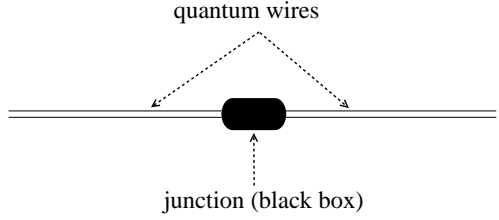}
  \caption{\scriptsize 
The unit of our quantum device consists of 
the two quantum wires and the junction as the black box.}
  \label{fig:unit}
\end{center}
\end{figure}
In this unit, the junction is for controlling 
the phase factor, and the wires for transporting 
the Schr\"{o}dinger particle. 
In our argument the two wires were 
respectively represented by the left island $\Omega_{\Lambda,L}$ 
and the right island $\Omega_{\Lambda,R}$, 
and the junction by the segment $[-\Lambda , +\Lambda]$.  
We regarded the junction as a black box 
to give mathematical, physical arbitrariness to the junction. 
As seen above, the self-adjointness of the Hamiltonian 
of the Schr\"{o}dinger particle living in the two wires 
$\Omega_{\Lambda}\equiv\Omega_{\Lambda,L}\cup\Omega_{\Lambda,R}$ 
is mathematically determined by a boundary condition 
of the wave functions on which every self-adjoint 
extension acts. 
In real physics, actually, the boundary condition 
is uniquely determined by the quality and the shape 
of the boundary of a material of the wires. 
Thus, the wave functions have to satisfy 
the unit's own specific boundary condition 
to become the residents of the unit,
otherwise the unit ejects them. 
We then consider the combination of different two units 
as in Fig.\ref{fig:units}. 
We set Unit$0$ with the boundary condition BC$0$ and 
Unit$1$ with the boundary condition BC$1$. 
We connect the two junctions with each other by 
a phase-controller. 
\begin{figure}[htbp]
\begin{center}
  \includegraphics[width=65mm]{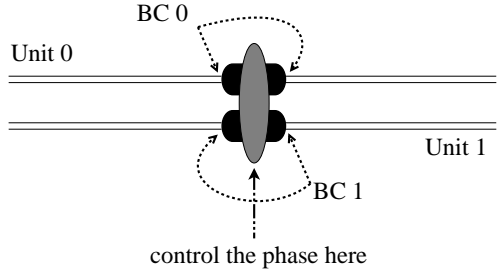}
  \caption{\rm{\scriptsize 
The upper is Unit$0$ with the boundary condition BC$0$, 
and the lower is Unit$1$ with the boundary condition BC$1$. 
The ellipsoid body is a phase-controller.}}
  \label{fig:units}
\end{center}
\end{figure}
We employ $\theta_{0}$ for the phase factor in BC$0$ and 
$\theta_{1}$ ($\ne\theta_{0}$) for the phase factor in BC$1$, respectively. 
If we succeed in the phase-shift gate operation 
from $\theta_{0}$ to $\theta_{1}$ at the phase-controller, 
we can make a phase-based switching device for electron 
as the Schr\"{o}dinger particle. 
This suggests that we could 
use Unit$0$ and Unit$1$ as a qubit 
in the same way as the flying qubit \cite{UniTokyo12} 
which is realized by the presence of an electron 
in either channel of the wires of an Aharonov-Bohm ring. 
Thus, our mathematical idea of the phase-based switching device 
may propose a possibility of the control over the transport of 
Schr\"{o}dinger particle so that 
the transport through either channel of Unit$0$ and Unit$1$ 
makes another flying qubit by using phase factor.

\section{Proof of Theorems}

\subsection{Proof of Theorem \ref{theo:main1}}
\label{subsec:theo-main1}

We here introduce the brief outline 
of the proof of Theorem \ref{theo:main1} 
according to Theorem 1(ii) and Theorem 2(i). 
We will use the following lemma: 
\begin{lemma}
\label{lem:2}
\begin{enumerate}
\item[i)] Let $a_{+}, a_{-}$ be arbitrary complex numbers.  
\begin{enumerate}
\item[(i-1)] For every $\rho=(\rho_{+},\rho_{-})$ 
with $|\rho_{\pm}|<\infty$, 
there is a wave function $\psi_{\rho}\in D(H_{\rho})$ 
such that $\psi_{\rho}(\pm\Lambda)=a_{\pm}^{*}$. 
\item[(i-2)] For every $\rho=(\rho_{+},\rho_{-})$ 
with $|\rho_{+}|<\infty$ and $\rho_{-}=+\infty$, 
there is a wave function $\psi_{\rho}\in D(H_{\rho})$ 
such that $\psi_{\rho}(+\Lambda)=a_{+}^{*}$ and 
$\psi_{\rho}'(-\Lambda)=a_{-}^{*}$. 
\item[(i-3)] For every $\rho=(\rho_{+},\rho_{-})$ 
with $\rho_{+}=+\infty$ and $|\rho_{-}|<\infty$, 
there is a wave function $\psi_{\rho}\in D(H_{\rho})$ 
such that $\psi_{\rho}(-\Lambda)=a_{-}^{*}$ and 
$\psi_{\rho}'(+\Lambda)=a_{+}^{*}$. 
\item[(i-4)] For every $\rho=(\rho_{+},\rho_{-})$ 
with $\rho_{\pm}=+\infty$, 
there is a wave function $\psi_{\rho}\in D(H_{\rho})$ 
such that $\psi_{\rho}'(\pm\Lambda)=a_{\pm}^{*}$. 
\end{enumerate}  
\item[ii)] For arbitrary complex numbers $a_{1}, a_{2}$, 
there is a wave function $\psi_{\alpha}\in D(H_{\alpha})$ 
such that $\psi_{\alpha}(-\Lambda)=a_{1}^{*}$ 
and $\psi_{\alpha}'(-\Lambda)~a_{2}^{*}$. 
\end{enumerate}
\end{lemma}

\begin{proof} We denote by $\chi_{L}(x)$ and $\chi_{R}(x)$ 
the characteristic functions on 
$\overline{\Omega_{\Lambda,L}}
=\left(\left. -\infty , -\Lambda\right]\right.$ 
and $\overline{\Omega_{\Lambda,R}}
=\left.\left[ +\Lambda , +\infty\right.\right)$, respectively. 

(i) Let $f, g$ be functions in $C_{0}^{\infty}(\mathbb{R})$ 
so that $f(0)\ne 0$, $f'(0)/f(0)=1$, 
and $g(0)=0$, $g'(0)\ne 0$. 
We define the function $\psi_{\rho}$ by 
$$
\psi_{\rho}(x)
:=\frac{a_{-}^{*}}{f(0)}f(\rho_{-}(x+\Lambda))\chi_{L}(x)
+\frac{a_{+}^{*}}{f(0)}f(\rho_{+}(x-\Lambda))\chi_{R}(x),\quad 
x \in \Omega_{\Lambda}, 
$$
for part (i-1), 
$$
\psi_{\rho}(x)
:=\frac{a_{-}^{*}}{g'(0)}g(x+\Lambda)\chi_{L}(x)
+\frac{a_{+}^{*}}{f(0)}f(\rho_{+}(x-\Lambda))\chi_{R}(x),\quad 
x \in \Omega_{\Lambda}, 
$$
for part (i-2), 
$$
\psi_{\rho}(x)
:=\frac{a_{-}^{*}}{f(0)}f(\rho_{-}(x+\Lambda))\chi_{L}(x)
+\frac{a_{+}^{*}}{g'(0)}g(x-\Lambda)\chi_{R}(x),\quad 
x \in \Omega_{\Lambda}, 
$$
for part (i-3), 
$$
\psi_{\rho}(x)
:=\frac{a_{-}^{*}}{g'(0)}g(x+\Lambda)\chi_{L}(x)
+\frac{a_{+}^{*}}{g'(0)}g(x-\Lambda)\chi_{R}(x),\quad 
x \in \Omega_{\Lambda}, 
$$
for part (i-4).  
Then, we obtain our desired function $\psi_{\rho}(x)$. 

(ii) Let $f, g$ be functions in $C_{0}^{\infty}(\mathbb{R})$ 
so that $f(0)\ne 0$, $f'(0)=0$, 
and $g(0)=0$, $g'(0)\ne 0$. 
We define the function $\psi_{\rho}$ by 
\begin{align*}
\psi_{\alpha}(x):=&
\left\{ 
\frac{a_{1}^{*}}{f(0)}f(x+\Lambda)
+\frac{a_{2}^{*}}{g'(0)}g(x+\Lambda)
\right\}\chi_{R}(x) \\ 
&\qquad 
+ \left\{ 
\frac{b_{1}}{f(0)}f(x-\Lambda)
+\frac{b_{2}}{g'(0)}g(x-\Lambda)
\right\}\chi_{R}(x),
\end{align*} 
where 
$b_{1}=\alpha_{1}a_{1}^{*}+\alpha_{2}a_{2}^{*}$ 
and $b_{2}=\alpha_{3}a_{1}^{*}+\alpha_{4}a_{2}^{*}$.  
Then, the function $\psi_{\alpha}$ satisfies our desired condition. 
\end{proof}

It follows from integration by parts that 
\begin{align}
\langle H_{0}^{*}\psi|\phi\rangle_{L^{2}(\Omega_{\Lambda})}
-\langle\psi|H_{0}^{*}\phi\rangle_{L^{2}(\Omega_{\Lambda})} 
=& 
\psi'(+\Lambda)^{*}\phi(+\Lambda)-\psi(+\Lambda)^{*}\phi'(+\Lambda) 
\label{eq:int-by-parts} \\  
&-\psi'(-\Lambda)^{*}\phi(-\Lambda)+\psi(-\Lambda)^{*}\phi'(-\Lambda) 
\notag 
\end{align}
for every $\psi, \phi \in D(H_{0}^{*})$. 

Since $H_{\sharp}\subset H_{0}^{*}$ for $\sharp=\rho, \alpha$, 
it is easy to check $H_{\sharp}\subset H_{\sharp}^{*}$ 
using (\ref{eq:int-by-parts}).  
On the other hand, $H_{\sharp}\supset H_{\sharp}^{*}$, $\sharp=\rho, \alpha$, 
is proved as follows:  
By the general definition of adjoint operator, 
we have 
$0=\langle H_{\sharp}\psi_{\sharp}|\phi\rangle_{L^{2}(\Omega_{\Lambda})}
-\langle \psi_{\sharp}|H_{\sharp}^{*}\phi\rangle_{L^{2}(\Omega_{\Lambda})}$ 
for every $\phi\in D(H_{\sharp}^{*})$ and $\psi_{\sharp}\in D(H_{\sharp})$ 
given in Lemma \ref{lem:2}. 
The arbitrariness of $a_{+}$ and $a_{-}$, or $a_{1}$ and $a_{2}$ 
implies $\phi\in D(H_{\sharp})$, 
that is, $H_{\sharp}^{*}\subset H_{\sharp}$. 
This is our path to prove the self-adjointness of 
$H_{\sharp}$, i.e., $H_{\sharp}^{*}=H_{\sharp}$, 
employed in \cite{FHNS10}.

For instance, we here show the proof of Theorem \ref{theo:main1} (b) only. 
Since $H_{\alpha}\subset H_{0}^{*}$, 
we have 
\begin{align*}
& \langle H_{\alpha}\psi|\phi\rangle_{L^{2}(\Omega_{\Lambda})}
-\langle \psi|H_{\alpha}\phi\rangle_{L^{2}(\Omega_{\Lambda})} \\ 
=&
\left(\alpha_{3}\psi(-\Lambda)+\alpha_{4}\psi'(-\Lambda)\right)^{*}
\left(\alpha_{1}\phi(-\Lambda)+\alpha_{2}\phi'(-\Lambda)\right) \\ 
&-\left(\alpha_{1}\psi(-\Lambda)+\alpha_{2}\psi'(-\Lambda)\right)^{*}
\left(\alpha_{3}\phi(-\Lambda)+\alpha_{4}\phi'(-\Lambda)\right) \\ 
&-\psi'(-\Lambda)^{*}\phi(-\Lambda)
+\psi(-\Lambda)^{*}\phi'(-\Lambda) \\ 
=& 
(\alpha_{1}\alpha_{3}^{*}-\alpha_{1}^{*}\alpha_{3})
\psi(-\Lambda)^{*}\phi(-\Lambda)
+(\alpha_{1}\alpha_{4}^{*}-\alpha_{2}^{*}\alpha_{3}-1)
\psi'(-\Lambda)^{*}\phi(-\Lambda) \\ 
&-(\alpha_{1}^{*}\alpha_{4}-\alpha_{2}\alpha_{3}^{*}-1)
\psi(-\Lambda)^{*}\phi'(-\Lambda) 
+(\alpha_{2}\alpha_{4}^{*}-\alpha_{2}^{*}\alpha_{4})
\psi'(-\Lambda)^{*}\phi'(-\Lambda) \\ 
=&0
\end{align*} 
for every $\psi, \phi \in D(H_{\alpha})$ 
by using (\ref{eq:int-by-parts}), 
(\ref{schrodinger:*alpha}) and conditions of (Class $\alpha$). 
Hence it follows from this that $H_{\alpha}$ is symmetric, 
i,e, $H_{\alpha}\subset H_{\alpha}^{*}$.

Conversely, using the fact that 
$H_{\alpha}\subset H_{\alpha}^{*}\subset H_{0}^{*}$ 
along with the help of the general definition of adjoint operator 
and (\ref{eq:int-by-parts}), 
for every $\phi\in D(H_{\alpha}^{*})$ and 
$\psi_{\alpha}\in D(H_{\alpha})$ given in Lemma \ref{lem:2}(ii) 
we have 
\begin{align*}
0=&\langle H_{\alpha}\psi_{\alpha}|\phi\rangle_{L^{2}(\Omega_{\Lambda})}
-\langle \psi_{\alpha}|H_{\alpha}^{*}\phi\rangle_{L^{2}(\Omega_{\Lambda})} \\ 
=& 
\left(
\alpha_{3}\psi_{\alpha}(-\Lambda)+\alpha_{4}\psi_{\alpha}'(-\Lambda)
\right)^{*}\phi(+\Lambda)
- \left(
\alpha_{1}\psi_{\alpha}(-\Lambda)+\alpha_{2}\psi_{\alpha}'(-\Lambda)
\right)^{*}\phi'(+\Lambda) \\ 
&-\psi'(-\Lambda)^{*}\phi(-\Lambda)
+\psi(-\Lambda)^{*}\phi'(-\Lambda) \\ 
=& 
a_{1}\left(
\alpha_{3}^{*}\phi(+\Lambda)-\alpha_{1}^{*}\phi'(+\Lambda)
+\phi'(-\Lambda)
\right)
+a_{2}\left(
\alpha_{4}^{*}\phi(+\Lambda)-\alpha_{2}^{*}\phi'(+\Lambda)
-\phi(-\Lambda)
\right).
\end{align*} 
So, the arbitrariness of $a_{1}$ and $a_{2}$ leads to 
$$
\begin{pmatrix}
\phi(-\Lambda) \\ 
\phi'(-\Lambda)
\end{pmatrix}
=
\begin{pmatrix}
\alpha_{4}^{*} & -\alpha_{2}^{*} \\ 
-\alpha_{3}^{*} & \alpha_{1}^{*}  
\end{pmatrix}
\begin{pmatrix}
\phi(+\Lambda) \\ 
\phi'(+\Lambda)
\end{pmatrix}.
$$
We here note that the conditions of 
(Class $\alpha$) leads to 
$$
B_{\alpha}^{-1}=\begin{pmatrix}
\alpha_{4}^{*} & -\alpha_{2}^{*} \\ 
-\alpha_{3}^{*} & \alpha_{1}^{*}  
\end{pmatrix}
$$
Thus, these two equations imply 
that $\phi\in D(H_{\alpha})$. 
That is, $H_{\alpha}^{*}\subset H_{\alpha}$. 
Therefore, we have proved that $H_{\alpha}^{*}=H_{\alpha}$.

\subsection{Proof of Theorem \ref{theo:main2}}

Part (a) follows from \cite[Theorem 1]{FHNS10}. 
So, we prove part (b) only. 

In this proof, we set $\eta:=e^{i\pi/4}$ 
for simplicity. 
Let us give an arbitrary non-diagonal $U\in U(2)$. 
We know that $U\in U(2)$ has the representation 
in Lemma \ref{lem:U(2)}. 

By (\ref{eq:von-Neumann-representation}) we can write 
$\psi\in D(H_{U})$ as 
$$
\psi=\psi_{0}+c_{L}L_{+}+c_{R}R_{+}
+c_{L}UL_{+}+c_{R}UR_{+}, 
$$ 
where $\psi\in D(H_{0})$, and 
$c_{L}$ and $c_{R}$ run over $\mathbb{C}$ arbitrarily. 
Using this representation, 
(\ref{eq:rel1}) and (\ref{eq:rel2}) we can compute 
$\psi(+\Lambda)$, $\psi'(+\Lambda)$, 
$\psi(-\Lambda)$, and $\psi'(-\Lambda)$ as: 
\begin{equation}
\begin{cases}
\psi(+\Lambda)=
-\gamma_{3}\gamma_{2}^{*}R_{+}(+\Lambda)^{*}c_{L}
+\left\{ R_{+}(+\Lambda)
+\gamma_{3}\gamma_{1}^{*}R_{+}(+\Lambda)^{*}\right\}c_{R}, \\ 
\psi'(+\Lambda)= 
e^{i\pi/4}\gamma_{3}\gamma_{2}^{*}R_{+}(+\Lambda)^{*}c_{L} \\ 
\qquad\qquad\qquad 
-\left\{ 
e^{-i\pi/4}R_{+}(+\Lambda)
+e^{i\pi/4}\gamma_{3}\gamma_{1}^{*}R_{+}(+\Lambda)^{*}
\right\}c_{R}, \\ 
\psi(-\Lambda)=
\left\{ 
R_{+}(+\Lambda)+\gamma_{3}\gamma_{1}R_{+}(+\Lambda)^{*}\right\}c_{L} 
+\gamma_{3}\gamma_{2}R_{+}(+\Lambda)^{*}c_{R}, \\ 
\psi'(-\Lambda)=
\left\{ e^{-i\pi/4}R_{+}(+\Lambda)
+e^{i\pi/4}\gamma_{3}\gamma_{1}R_{+}(+\Lambda)^{*}\right\}c_{L} \\ 
\qquad\qquad\qquad 
+e^{i\pi/4}\gamma_{3}\gamma_{2}R_{+}(+\Lambda)^{*}c_{R}.  
\end{cases}
\label{eq:representation-L-R}
\end{equation}
(\ref{eq:representation-L-R}) says that 
$$
\begin{pmatrix}
\psi(+\Lambda) \\ 
\psi'(+\Lambda)
\end{pmatrix}
=A_{+}
\begin{pmatrix}
c_{L} \\ 
c_{R}
\end{pmatrix}\quad 
\text{and}\quad 
\begin{pmatrix}
\psi(-\Lambda) \\ 
\psi'(-\Lambda)
\end{pmatrix}
=A_{-}
\begin{pmatrix}
c_{L} \\ 
c_{R}
\end{pmatrix},  
$$
where 
$$
A_{+}=
\begin{pmatrix}
-\gamma_{3}\gamma_{2}^{*}R_{+}(+\Lambda)^{*} & 
R_{+}(+\Lambda)+\gamma_{3}\gamma_{1}^{*}R_{+}(+\Lambda)^{*} \\ 
\eta\gamma_{3}\gamma_{2}^{*}R_{+}(+\Lambda)^{*} & 
-\left(\eta^{*}R_{+}(+\Lambda)
+\eta\gamma_{3}\gamma_{1}^{*}R_{+}(+\Lambda)^{*}\right)
\end{pmatrix}
$$
and 
$$
A_{-}=
\begin{pmatrix}
R_{+}(+\Lambda)+\gamma_{3}\gamma_{1}R_{+}(+\Lambda)^{*} & 
\gamma_{3}\gamma_{2}R_{+}(+\Lambda)^{*} \\ 
\eta^{*}R_{+}(+\Lambda)+\gamma_{3}\gamma_{1}R_{+}(+\Lambda)^{*} & 
\gamma_{3}\gamma_{2}R_{+}(+\Lambda)^{*}
\end{pmatrix}.
$$
Since $\det A_{-}=i\sqrt{2}|R_{+}(+\Lambda)|^{2}\gamma_{3}\gamma_{2}\ne 0$, 
we know $A_{-}^{-1}$ exists.  
Thus, our desired representation 
(\ref{eq:TJF-B}) of $\alpha_{1}, \alpha_{2}, \alpha_{3}$, 
and $\alpha_{4}$ by $\gamma_{1}, \gamma_{2}$, and $\gamma_{3}$ 
follows from the immediate computation of $B_{\alpha}=A_{+}A_{-}^{-1}$.  
Thus, every $\psi\in D(H_{U})$ satisfies the boundary condition 
(BC $\alpha$). 
What we have to show next is that the 
vector $\alpha\in\mathbb{C}^{4}$ given by (\ref{eq:TJF-B}) 
is in the class (Class $\alpha$). 
It is obvious that our $\alpha_{1}$, $\alpha_{2}$, $\alpha_{3}$, 
and $\alpha_{4}$ satisfy (\ref{eq:lem-3-2}). 
We can compute $\alpha_{1}\alpha_{4}^{*}-\alpha_{2}\alpha_{3}^{*}$ 
as follows: 
\begin{align*}
\alpha_{1}\alpha_{4}^{*}-\alpha_{2}\alpha_{3}^{*}  
=&
\frac{2}{|\gamma_{2}|}
\Bigl[
\left(
\Re(\eta\gamma_{1})+\Re(\eta e^{-i\sqrt{2}\,\Lambda}\gamma_{3})
\right)
\left(
\Re(\eta^{*}\gamma_{1})+\Re(\eta e^{-i\sqrt{2}\,\Lambda}\gamma_{3})
\right) \\ 
&\qquad  
-\left( 
\Re\gamma_{1}+\Re(e^{-i\sqrt{2}\,\Lambda}\gamma_{3})
\right)
\left( 
\Re\gamma_{1}+\Re(\eta^{2}e^{-i\sqrt{2}\,\Lambda}\gamma_{3})
\right)
\Bigr] \\ 
=&
\frac{1}{2|\gamma_{2}|}
\Bigl[
\left(
\eta\gamma_{1}+\eta^{*}\gamma_{1}^{*}
+\eta e^{-i\sqrt{2}\,\Lambda}\gamma_{3}+\eta^{*}e^{i\sqrt{2}\,\Lambda}\gamma_{3}^{*}
\right) \\
&\qquad\qquad\times
\left(
\eta^{*}\gamma_{1}+\eta\gamma_{1}^{*}
+\eta e^{-i\sqrt{2}\,\Lambda}\gamma_{3}
+\eta^{*}e^{i\sqrt{2}\,\Lambda}\gamma_{3}^{*}
\right) \\ 
&\qquad  
-\left( 
\gamma_{1}+\gamma_{1}^{*}
+e^{-i\sqrt{2}\,\Lambda}\gamma_{3}
+e^{i\sqrt{2}\,\Lambda}\gamma_{3}^{*}
\right) \\
&\qquad\qquad\times 
\left( 
\gamma_{1}+\gamma_{1}^{*}
+\eta^{2}e^{-i\sqrt{2}\,\Lambda}\gamma_{3}
+\eta^{*2}e^{i\sqrt{2}\,\Lambda}\gamma_{3}^{*}
\right)
\Bigr] \\ 
=& 
\frac{|\gamma_{3}|^{2}-|\gamma_{1}|^{2}}{|\gamma_{2}|^{2}}=1. 
\end{align*}
Thus, $\alpha_{1}$, $\alpha_{2}$, $\alpha_{3}$, 
and $\alpha_{4}$ given by (\ref{eq:TJF-B}) 
satisfy (\ref{eq:lem-3-1}), 
and the vector $\alpha=(\alpha_{1}, \alpha_{2}, 
\alpha_{3}, \alpha_{4})$ is in the class (Class $\alpha$). 
 
Therefore, we have constructed the boundary matrix 
$B_{\alpha}$ with the vector $\alpha\in\mathbb{C}^{4}$ 
in the class (Class $\alpha$) 
from every non-diagonal $U\in U(2)$.

Conversely, let $\alpha_{1}, \alpha_{2}, \alpha_{3}$, and 
$\alpha_{4}$ be arbitrary complex numbers in 
the class (Class $\alpha$). 
It immediately follows from the definition 
of $\gamma_{1}$, $\gamma_{2}$, and $\Gamma_{0}$ 
in (\ref{eq:TJF-U}) that 
$|\gamma_{1}|^{2}+|\gamma_{2}|^{2}=1$. 
Using this equation together with 
the conditions of (Class $\alpha$) 
and the representation given in 
Proposition \ref{prop:phase}, 
we have 
\begin{align*}
1=&|\gamma_{1}|^{2}+|\gamma_{2}|^{2}
=\Gamma_{0}^{2}
\left\{ 
\sum_{j=1}^{4}|a_{j}|^{2}
+\sqrt{2}
\left( 
a_{1}a_{2}+a_{1}a_{3}+a_{2}a_{4}+a_{3}a_{4}
\right)
+2a_{1}a_{4}
\right\} \\ 
=&|\gamma_{3}|^{2}. 
\end{align*}
Here we note $\gamma_{2}\ne 0$ by its definition. 
Thus, (\ref{eq:TJF-U}) gives us the unitary operator $U$ 
with the representation: 
$$
U=
\gamma_{3}
\begin{pmatrix}
\gamma_{1} & -\gamma_{2}^{*} \\ 
\gamma_{2} & \gamma_{1}^{*}
\end{pmatrix}
\in U(1)S\mathbb{H}=U(2).
$$
We show from now on that the above $U$ satisfies 
$D(H_{U})=D(H_{\alpha})$. 

For arbitrarily given $U\in U(2)$, 
insert $\psi(+\Lambda)$, $\psi'(+\Lambda)$, 
$\psi(-\Lambda)$, and $\psi'(-\Lambda)$ 
with the representation (\ref{eq:representation-L-R}) 
into the boundary conditions, 
$$
\begin{cases}
\psi(+\Lambda)=\alpha_{1}\psi(-\Lambda)+\alpha_{2}\psi'(-\Lambda), \\ 
\psi'(+\Lambda)=\alpha_{3}\psi(-\Lambda)+\alpha_{4}\psi'(-\Lambda).  
\end{cases}
$$
Then, since the coefficients $c_{L}$ and $c_{R}$ 
in $D(H_{U})$ are arbitrary and $\gamma_{3}^{-1}=\gamma_{3}^{*}$,  
we can show that the condition $D(H_{U})=D(H_{\alpha})$ 
is equivalent to the system of the following system of equations: 
\begin{align}
& \left(\alpha_{1}+e^{i\pi/4}\alpha_{2}\right)\gamma_{2}-\gamma_{1}^{*}
=e^{i\sqrt{2}\, \Lambda}\gamma_{3}^{*},
\label{eq:inverse-1} \\  
& \gamma_{2}^{*}+\left(\alpha_{1}+e^{i\pi/4}\alpha_{2}\right)\gamma_{1}
=\, -\left(\alpha_{1}+e^{-i\pi/4}\alpha_{2}\right)e^{i\sqrt{2}\, \Lambda}\gamma_{3}^{*}, 
\label{eq:inverse-2} \\  
& e^{i\pi/4}\gamma_{2}^{*}
-\left(\alpha_{3}+e^{i\pi/4}\alpha_{4}\right)\gamma_{1}
=\left(\alpha_{3}+e^{-i\pi/4}\alpha_{4}\right)e^{i\sqrt{2}\, \Lambda}\gamma_{3}^{*}, 
\label{eq:inverse-3} \\  
& \left(\alpha_{3}+e^{i\pi/4}\alpha_{4}\right)\gamma_{2}
+e^{i\pi/4}\gamma_{1}^{*}
=\, -e^{i(\sqrt{2}\, \Lambda-\pi/4)}\gamma_{3}^{*}.
\label{eq:inverse-4}
\end{align} 

Thus, we now show that our $U$ given by (\ref{eq:TJF-U}) 
satisfies the system of equations: 
Using the representation in Proposition \ref{prop:phase}, 
it is easy to check that our $\gamma_{1}$, $\gamma_{2}$, 
and $\gamma_{3}$ given by (\ref{eq:TJF-U}) 
satisfy (\ref{eq:inverse-1}) and 
(\ref{eq:inverse-4}) in the following. 
\begin{align*}
& \left(\alpha_{1}+\eta\alpha_{2}\right)\gamma_{2}
-\gamma_{1}^{*} 
=\, -\Gamma_{0}\left(\eta a_{1}+ia_{2}+a_{3}+\eta a_{4}\right)
=e^{i\sqrt{2}\, \Lambda}\gamma_{3}^{*}, \\ 
& \left(\alpha_{3}+\eta\alpha_{4}\right)\gamma_{2}
+\eta\gamma_{1}^{*} 
=\Gamma_{0}\left( a_{1}+\eta a_{2}+\eta^{*}a_{3}+a_{4}\right)
=\, -\eta^{*}e^{i\sqrt{2}\, \Lambda}\gamma_{3}^{*}. 
\end{align*}
We recall $a_{1}a_{4}-a_{2}a_{3}=1$ by Proposition \ref{prop:phase}. 
This equation leads to 
\begin{align}
& \eta a_{1}a_{4}+i\eta^{*}a_{2}^{2}+\eta^{*}a_{2}a_{3}
=\eta^{*}a_{1}a_{4}+\eta a_{2}^{2}+\eta a_{2}a_{3}+i\sqrt{2}
\label{eq:inverse-5}, \\ 
& a_{1}a_{4}+ia_{2}a_{3}=ia_{1}a_{4}+a_{2}a_{3}+1-i.
\label{eq:inverse-6} 
\end{align}
We can show that our $\gamma_{1}$, $\gamma_{2}$, 
and $\gamma_{3}$ satisfy (\ref{eq:inverse-2}) 
as 
\begin{align*}
\gamma_{2}^{*}+\left(\alpha_{1}+\eta\alpha_{2}\right)\gamma_{1}
=&e^{i\theta}\Gamma_{0}
\Bigl\{ 
\eta a_{1}^{2}+(1+i)a_{1}a_{2}+a_{1}a_{3}+\eta^{*}a_{1}a_{4} \\ 
&\qquad\quad  
+\eta a_{2}^{2}+\eta a_{2}a_{3}+a_{2}a_{4}+i\sqrt{2}
\Bigr\} \\ 
=&\, -\left(\alpha_{1}+\eta^{*}\alpha_{3}\right)
e^{i\sqrt{2}\, \Lambda}\gamma_{3}^{*}
\end{align*} 
with the help of (\ref{eq:inverse-5}), and 
they satisfy   
(\ref{eq:inverse-2}) 
as 
\begin{align*}
\eta\gamma_{2}^{*}-\left(\alpha_{3}+\eta\alpha_{4}\right)\gamma_{1}
=&\, -e^{i\theta}\Gamma_{0}
\Bigl\{ 
\eta a_{1}a_{3}+ia_{1}a_{4}+a_{2}a_{3}+\eta a_{2}a_{4} \\ 
&\qquad\qquad 
+a_{3}^{2}+\sqrt{2}a_{3}a_{4}+a_{4}^{2}+1-i
\Bigr\} \\ 
=&\left(\alpha_{3}+\eta^{*}\alpha_{4}\right)
e^{i\sqrt{2}\, \Lambda}\gamma_{3}^{*}
\end{align*} 
with the help of (\ref{eq:inverse-6}). 

Therefore, the unitary operator $U$ made from 
our $\gamma_{1}$, $\gamma_{2}$, and 
$\gamma_{3}$ satisfies the equation $D(H_{U})=D(H_{\alpha})$, 
and we can complete the proof of our theorem.

\section{Conclusion}
We have completely characterized the boundary conditions 
for all the self-adjoint extensions 
of the minimal Schr\"{o}dinger operator. 
We then found a tunnel-junction formula 
concerning the phase factor. 
In this formula we can find 
the factor $\pm\sqrt{2}\,\Lambda$ 
which depends on the length of the junction. 
We have not yet clarified the physical reason 
why this factor appears. 
Compared with the results in \cite{HK1} and this paper, 
we realize that such a factor 
concerning the parameter $\Lambda$ 
appears for the Schr\"{o}dinger particle, 
but it does not for the Dirac particle \cite{HK1}. 
That is, the Schr\"{o}dinger particle seems 
to feel the distance $2\Lambda$, 
though the Dirac particle does not. 
We conjecture that this situation physically depends on 
the speed of the particle.

\section*{Acknowledgments}
One of the authors (M.H.) acknowledges the financial support from JSPS, 
Grant-in-Aid for Scientific Research (C) 23540204. 
He also expresses special thanks to 
Kae Nemoto and Yutaka Shikano for the useful discussions 
with them.

\end{document}